\documentclass[a4paper,UKenglish,numberwithinsect,
        thm-restate]{lipics-v2021}

\usepackage[utf8]{inputenc}
\usepackage{mathtools}
\usepackage{amsthm,amsmath}
\usepackage{mathrsfs}
\usepackage{hyperref}
\usepackage{xspace}

\usepackage[hyphenbreaks]{breakurl}
\usepackage[all]{hypcap}
\usepackage{stmaryrd}
\usepackage{graphicx}
\usepackage{keycommand}
\usepackage{multicol}
\usepackage{array}

\usepackage{tikzit}

\tikzstyle{box}=[shape=rectangle, text height=1.5ex, text depth=0.25ex, yshift=0.5mm, fill=white, draw=black, minimum height=5mm, yshift=-0.5mm, minimum width=5mm, font={\small}]
\tikzstyle{gate}=[shape=rectangle, text height=1.5ex, text depth=0.25ex, yshift=0.5mm, fill=white, draw=black, minimum height=5mm, yshift=-0.5mm, minimum width=5mm, font={\small}, tikzit category=circuit]
\tikzstyle{big gate}=[shape=rectangle, text height=1.5ex, text depth=0.25ex, yshift=0.5mm, fill=white, draw=black, minimum height=10mm, yshift=-0.5mm, minimum width=5mm, font={\small}, tikzit category=circuit]
\tikzstyle{Z dot}=[inner sep=0mm, minimum size=2mm, shape=circle, draw=black, fill={rgb,255: red,221; green,255; blue,221}, tikzit category=zx]
\tikzstyle{Z phase dot}=[minimum size=5mm, font={\footnotesize\boldmath}, shape=rectangle, rounded corners=2mm, inner sep=0.2mm, outer sep=-2mm, scale=0.8, tikzit shape=circle, draw=black, fill={rgb,255: red,221; green,255; blue,221}, tikzit draw=blue, tikzit category=zx]
\tikzstyle{X dot}=[Z dot, shape=circle, draw=black, fill={rgb,255: red,255; green,136; blue,136}, tikzit category=zx]
\tikzstyle{X phase dot}=[Z phase dot, tikzit shape=circle, tikzit draw=blue, fill={rgb,255: red,255; green,136; blue,136}, font={\footnotesize\boldmath}, tikzit category=zx]
\tikzstyle{hadamard}=[fill=yellow, draw=black, shape=rectangle, inner sep=0.6mm, minimum height=1.5mm, minimum width=1.5mm, tikzit category=zx]
\tikzstyle{paulibox}=[fill={rgb,255: red,221; green,221; blue,255}, draw=black, shape=rectangle, inner sep=0.6mm, minimum height=5mm, minimum width=5mm, font={\footnotesize}, text height=1.5ex, text depth=0.25ex, tikzit category=zx]
\tikzstyle{vertex}=[inner sep=0mm, minimum size=1mm, shape=circle, draw=black, fill=black, tikzit category=misc]
\tikzstyle{vertex set}=[inner sep=0mm, minimum size=1mm, shape=circle, draw=black, fill=white, font={\footnotesize\boldmath}, tikzit category=misc]
\tikzstyle{small black dot}=[fill=black, draw=black, shape=circle, inner sep=0pt, minimum width=1.2mm, tikzit category=circuit]
\tikzstyle{cnot ctrl}=[fill=black, draw=black, shape=circle, inner sep=0pt, minimum width=1.2mm, tikzit category=circuit]
\tikzstyle{cnot targ}=[fill=white, draw=white, shape=circle, tikzit category=circuit, label={center:$\oplus$}, inner sep=0pt, minimum width=2.1mm, tikzit fill={rgb,255: red,102; green,204; blue,255}, tikzit draw=black]
\tikzstyle{ket}=[fill=white, draw=black, shape=regular polygon, regular polygon sides=3, regular polygon rotate=-30, scale=0.7, inner sep=1pt, tikzit category=circuit, tikzit shape=rectangle, tikzit fill=green]
\tikzstyle{bra}=[fill=white, draw=black, shape=regular polygon, regular polygon sides=3, regular polygon rotate=30, scale=0.7, inner sep=1pt, tikzit category=circuit, tikzit shape=rectangle, tikzit fill=red]
\tikzstyle{scalar}=[shape=rectangle, text height=1.5ex, text depth=0.25ex, yshift=0.5mm, fill=white, draw=black, minimum height=5mm, yshift=-0.5mm, minimum width=5mm, font={\small}]
\tikzstyle{clabel}=[fill=white, draw=none, shape=rectangle, tikzit fill={rgb,255: red,56; green,255; blue,242}, font={\footnotesize}, inner sep=1pt, tikzit category=labels]
\tikzstyle{empty diagram}=[draw={gray!40!white}, dashed, shape=rectangle, minimum width=1cm, minimum height=1cm, tikzit category=misc]

\tikzstyle{simple}=[-]
\tikzstyle{hadamard edge}=[-, dashed, dash pattern=on 2pt off 0.5pt, thick, draw={rgb,255: red,68; green,136; blue,255}]
\tikzstyle{box edge}=[-, dashed, dash pattern=on 2pt off 0.5pt, thick, draw={rgb,255: red,203; green,192; blue,225}]
\tikzstyle{brace edge}=[-, tikzit draw=blue, decorate, decoration={brace,amplitude=1mm,raise=-1mm}]
\tikzstyle{diredge}=[->]
\tikzstyle{double edge}=[-, double, shorten <=-1mm, shorten >=-1mm, double distance=2pt]
\tikzstyle{gray edge}=[-, {gray!60!white}]
\tikzstyle{pointer edge}=[->, very thick, gray]
\tikzstyle{boldedge}=[-, line width=1.6pt, shorten <=-0.17mm, shorten >=-0.17mm]

\input{quantum.tikzdefs} 

\usetikzlibrary{calc}

\usepackage[all]{hypcap} 

\makeatletter
\newcommand\etc{etc\@ifnextchar.{}{.\@}\xspace}

\makeatother

\renewcommand{\P}{\textbf{\upshape P}\xspace}

\newcommand{\NP}{\textbf{\upshape NP}\xspace}
\newcommand{\sharpP}{\textbf{\upshape\texttt\#P}\xspace}
\newcommand{\PostBQP}{\textbf{\upshape PostBQP}\xspace}
\newcommand{\BQP}{\textbf{\upshape BQP}\xspace}
\newcommand{\PP}{\textbf{\upshape PP}\xspace}

\newcommand{\PH}{\textbf{\upshape PH}\xspace}

\newcommand{\CE}{\textbf{\upshape CircuitExtraction}\xspace}
\newcommand{\ACE}{\textbf{\upshape ApproxCircuitExtraction}\xspace}

\newcommand{\numberSAT}{\textbf{\upshape \texttt\#SAT}\xspace}

\hideLIPIcs

\title{Circuit Extraction for ZX-diagrams can be \sharpP-hard}

\author{Niel de Beaudrap}{University of Sussex, United Kingdom}{niel.debeaudrap@gmail.com}{https://orcid.org/
0000-0001-9549-5146}{}
\author{Aleks Kissinger}{University of Oxford, United Kingdom}{aleks.kissinger@cs.ox.ac.uk}{https://orcid.org/0000-0002-6090-9684}{Acknowledges support from AFOSR grant FA2386-18-1-4028.}
\author{John van de Wetering}{Radboud University Nijmegen, The Netherlands \and University of Oxford, United Kingdom \and \url{https://vdwetering.name}}{john@vdwetering.name}{https://orcid.org/0000-0002-5405-8959}{Acknowledges support from a NWO Rubicon Personal Grant.}

\authorrunning{N.~de Beaudrap, A.~Kissinger \& J.~van de Wetering} 

\Copyright{Niel de Beaudrap, Aleks Kissinger \& John van de Wetering} 

\ccsdesc[500]{Theory of computation~Quantum computation theory}

\keywords{ZX-calculus, circuit extraction, quantum circuits, \sharpP} 

\begin{document}
\nolinenumbers
\maketitle

\begin{abstract}
	The ZX-calculus is a graphical language for reasoning about quantum computation using ZX-diagrams, a certain flexible generalisation of quantum circuits that can be used to represent linear maps from $m$ to $n$ qubits for any $m,n \geq 0$. Some applications for the ZX-calculus, such as quantum circuit optimisation and synthesis, rely on being able to efficiently translate a ZX-diagram back into a quantum circuit of comparable size. While several sufficient conditions are known for describing families of ZX-diagrams that can be efficiently transformed back into circuits, it has previously been conjectured that the general problem of \emph{circuit extraction} is hard. That is, that it should not be possible to efficiently convert an arbitrary ZX-diagram describing a unitary linear map into an equivalent quantum circuit.
  In this paper we prove this conjecture by showing that the circuit extraction problem is \sharpP-hard, and so is itself at least as hard as strong simulation of quantum circuits.
  In addition to our main hardness result, which relies specifically on the circuit representation, we give a representation-agnostic hardness result. Namely, we show that any oracle that takes as input a ZX-diagram description of a unitary and produces samples of the output of the associated quantum computation enables efficient probabilistic solutions to NP-complete problems.
\end{abstract}

\section{Introduction}

Quantum circuit notation is widely used in the field of quantum computing to denote computations to be executed on a quantum computer. 
While quantum circuits are a useful tool for representing computations on a quantum computer, they are somewhat inconvenient for reasoning about computations (such as proving equalities or doing simplifications); and for representing computations in alternative models like the one-way model of measurement-based quantum computation (MBQC)~\cite{MBQC1}, or surface code lattice surgery~\cite{Horsman2012surgery}.

ZX-diagrams are an alternative, more general representation of quantum computations, which allow complex operations to be described using a few simple generating operators.
ZX-diagrams come with an equational theory, called the \emph{ZX-calculus}~\cite{CD1}, which allows one to perform many useful calculations graphically, without resorting to concrete matrix computations.
While ZX-diagrams can be seen as an extension of circuits~\cite{CD2}, they also readily admit encodings of the one-way model~\cite{DP2} and lattice surgery~\cite{horsman2017surgery}, and allow one to reason more easily about such procedures.
There are several known \emph{complete} axiomatisations of the ZX-calculus~\cite{HarnyCompleteness,vilmarteulercompleteness}, where any true equality of linear maps can be proved graphically.
For a review on the ZX-calculus we refer to~\cite{vandewetering2020zxcalculus}.

The ZX-calculus has been used in a variety of areas. It was used to optimise T-count~\cite{kissinger2019tcount,de2019techniques_old}, braided circuits~\cite{hanks2019effective} and MBQC~\cite{Backens2020extraction}; to find a new normal form for Clifford circuits~\cite{cliffsimp}; to do more effective classical simulation using stabiliser decompositions~\cite{kissinger2021simulating}; and to reason about surface codes~\cite{Gidney2019efficientmagicstate,autoCCZ}, mixed-state quantum computations~\cite{carette_completeness_2019}, natural language processing~\cite{coecke2020foundations}, condensed matter systems~\cite{east2020akltstates}, counting problems~\cite{debeaudrap2020tensor,townsend-teague2021classifying} and spin-networks~\cite{d.p.east2021spinnetworks}.

As a strict extension of quantum circuit language, ZX-diagrams may express operations in a form that do not correspond directly to a quantum circuit.
This added flexibility makes it easier to find novel strategies to simplify quantum circuits, but it comes at a cost: given a ZX-diagram representing a unitary linear map, it might be non-trivial to transform it back into a circuit of comparable size. Such a translation might however be necessary if, for instance, we want to run the computation described by a ZX-diagram on a gate-based quantum computer.

We refer to the above problem, as the \emph{circuit extraction} problem: given a ZX-diagram which denotes a unitary operator $U$, find a unitary circuit (\emph{i.e.},~a~quantum circuit without measurements) that implements $U$.
In recent years, some progress has been made on this problem~\cite{cliffsimp,kissinger2019tcount,Backens2020extraction,Simmons2021Measurement,KissingerCNOT2019,deBeaudrap2020Paulifusion}. However, all known methods for efficient extraction of circuits from ZX-diagrams rely on additional conditions, in particular requiring there to be some kind of \emph{flow} on the diagram, a concept imported from MBQC~\cite{GFlow}.
Such conditions allow the diagram to be rewritten incrementally into a unitary circuit. Since many ZX-calculus rewrites preserve these conditions, it is possible to perform optimisation of quantum circuits using ZX-calculus rules and still recover circuits efficiently.

However, it is worth trying to generalise these conditions as much as possible, or even remove them. For instance, it was noted in~\cite{kissinger2019tcount} that a certain transformation of ZX-diagrams would decrease the T-count (an important metric for quantum circuit optimisation), but in the process broke the invariant (the existence of a gflow), preventing a circuit from being extracted efficiently using known techniques. Given all this it is then natural to wonder about the following question:

\begin{center}
	\itshape
	Is there some efficient procedure to translate any \\ unitary ZX-diagram into a quantum circuit?
\end{center}

In this paper we present strong evidence that there is no such efficient procedure, by showing that the circuit extraction problem is \sharpP-hard in the worst case.
The complexity class \sharpP contains for instance the problem of strong simulation of quantum circuits, and counting the number of satisfying solutions to a Boolean formula, so \sharpP-hard problems are expected to be intractable. 
We prove \sharpP-hardness by giving an encoding of Boolean formulae into unitary ZX-diagrams in such a way that extracting a polysize circuit provides a solution to the associated \numberSAT instance.
A consequence of our result is that if there were a polynomial time algorithm for circuit extraction, then $\P = \NP$.

Alternatively, since there is an evident translation from a ZX-diagram into a quantum circuit with postselection, this result can equivalently be seen as expressing the hardness of translating a postselected circuit that is promised to be proportional to a unitary into a circuit without postselection. While intuitively this seems likely to be hard, particularly in light of Aaronson's landmark result that $\PostBQP = \PP$~\cite{aaronson2005quantum}, our hardness result seems to be quite different in nature due to the unitarity promise. In particular, the postselection does not seem to be the `source of power' in our proof: the measurement outcomes corresponding to the post-selections in our circuits occur with some bounded probability, independent of the problem size.

One could ask how much our hardness result is tied to the fact that we require a procedure that produces quantum circuits from ZX-diagrams. Especially, when considering that in most cases we are not interested in the circuit itself, but instead we simply want to sample the output of the quantum computation. Perhaps one could find some other procedure to ``program'' a quantum computer using a ZX-diagram describing a unitary and obtain samples of measurement outcomes. We show that an efficient such procedure is unlikely to exist for arbitrary ZX-diagrams, by finding that such a procedure allows you to probabilistically solve \NP-hard problems. So if there were some way to generically translate unitary ZX-diagrams into procedures which could be realised in polynomial time on a quantum computer, it would follow that the entire polynomial hierarchy is in \BQP, and in particular that $\NP \subseteq \BQP$.

The paper is structured as follows. We start by covering preliminaries on quantum circuits, ZX-diagrams and the necessary complexity theory in Section~\ref{sec:prelim}.
Then in Section~\ref{sec:circuitExtraction} we formally define the circuit extraction problem and prove it is hard. Section~\ref{sec:variations} considers several variations on circuit extraction, and in Section~\ref{sec:upperbound} we find some upper bounds on the hardness of circuit extraction. We end with some concluding remarks in Section~\ref{sec:conclusion}.

\section{Preliminaries}\label{sec:prelim}

\subsection{Quantum circuits}
\label{sec:circuits}

Since we wish to extract `a circuit' from a \zxdiagram, it will be helpful to first consider what we actually mean by a circuit.

In quantum computational theory, a `circuit' is a description of a computational process consisting of operations which may be decomposed as a sequence of primitive `gates' and `measurements', which act on one or more qubits to change the states of those qubits.
The state-space of a qubit is identified with the finite-dimensional Hilbert space $\mathcal H_2 \cong \mathbb C^2$; the state of $k$ qubits in parallel is described by the tensor product $\mathcal H_2^{\otimes k}$.
A `gate' is an operation which is applied to one or more qubits and implements a unitary transformation $U: \mathcal H_2^{\otimes k} \to \mathcal H_2^{\otimes k}$ on the associated state space. 
A `measurement' is an operation which transforms a state $\ket{\psi} \in \mathcal H_2^{\otimes k}$ to some state $p_j^{-1/2} \,\Pi_j \ket{\psi}$ where $\{ \Pi_1, \Pi_2, \ldots \}$ is a set of projections that sum up to the identity operator $I$, the $p_j$ gives the probability of observing that particular measurement outcome and is given by $p_j = \bra{\psi} \Pi_j \ket{\psi}$, and the index $j$ provides the classical `outcome' indicating which transformation occurred.
A gate or measurement acting on a small number of qubits can be applied to a larger set of qubits by taking the tensor product with an appropriate number of identity operators.
A~`circuit' is then a composition of such gates and measurements on some number of qubits, acting in sequence or in parallel, to describe more complex (and in general, non-deterministic and irreversible) transformations of a quantum state-space.
For the purposes of using quantum circuits to define a reasonable model of computation, one usually elaborates the above with a description of how one would specify a circuit as part of a family of unitary operators, acting on inputs of various sizes (see Appendix~\ref{apx:circuits} for details).
For our purposes, it will suffice to require that the coefficients of the gates be efficiently computable, and in particular provided explicitly in some representation which suffices to approximate them to $O(\text{poly}(n))$ bits of precision in time $O(\mathrm{poly}(n))$ for an $n$ qubit circuit.

It will be convenient to refer to one specific such gate-set --- an infinite set $\mathcal B$ of gates, consisting of the single-qubit gates $Z_\alpha$ for arbitrary angles $\alpha$, the single-qubit Hadamard gate $H$ and the two-qubit gate $\CNOT$:
\begin{equation}
  Z_\alpha \ = \ \begin{pmatrix}
  1 & 0 \\ 0 & \mathrm e^{i\alpha}
  \end{pmatrix}
  \qquad
  H \ = \ \frac{1}{\sqrt{2}}\begin{pmatrix}
    1 & 1 \\ 1 & -1
  \end{pmatrix}
  \qquad
  \CNOT
  \ = \
  \text{\footnotesize$
  \begin{pmatrix}
    \,1 & 0 & 0 & 0\; \\
    \,0 & 1 & 0 & 0\; \\
    \,0 & 0 & 0 & 1\; \\
    \,0 & 0 & 1 & 0\; 
  \end{pmatrix}
  $}.
\end{equation}
This gate set forms a \emph{universal} gate set, meaning that a unitary acting on any number of qubits can be written as a circuit consisting of these gates~\cite{NielsenChuang}.
Other universal gate-sets exist, but so long as one considers gate sets whose parameters are efficiently computable from some input parameters and which act only on a bounded numbers of qubits (\emph{e.g.},~at most two or three qubits), the size of a circuit to represent a given unitary operator can only vary by a constant factor, so that for the purpose of complexity theory, the details of the specific gate set chosen are not important.

A circuit which contains no measurements, and therefore consists entirely of unitary gates, is called a `unitary circuit'.
A unitary circuit is reversible, and `deterministic' in the sense that an idealised realisation of such a circuit will transform the state-space in the same way each time.
As this is a convenient feature for the design and analysis of quantum algorithms, much of the literature on quantum algorithms concerns itself with unitary circuits, and much of the design of quantum computers is concerned with how to reliably implement unitary circuits.

\subsection{ZX-diagrams}

We provide a brief overview of ZX-diagrams.
For a review see~\cite{vandewetering2020zxcalculus}, and for a book-length introduction see Ref.~\cite{CKbook}.

ZX-diagrams form a diagrammatic language similar to the familiar
quantum circuit notation.  A \emph{\zxdiagram} (or simply
\emph{diagram}) consists of \emph{wires} and \emph{spiders}.  Wires
entering the diagram from the left are \emph{inputs}; wires exiting to
the right are \emph{outputs}.  Given two diagrams we can compose them
by joining the outputs of the first to the inputs of the second, or
form their tensor product by simply stacking the two diagrams~\cite{CD1,CD2}.

\emph{Spiders} are linear operations which can have any number of input or output
wires.  There are two varieties: $Z$-spiders depicted as green dots and $X$-spiders depicted as red dots:
\begin{equation}
\label{eqn:spiders}
\begin{aligned}
\begin{tikzpicture}
	\begin{pgfonlayer}{nodelayer}
		\node [style=Z phase dot] (0) at (0, 0) {$\alpha$};
		\node [style=none] (1) at (1.25, 1) {};
		\node [style=none] (2) at (-1.25, 1) {};
		\node [style=none] (3) at (-1.25, -1) {};
		\node [style=none] (4) at (1.25, -1) {};
		\node [style=none] (5) at (1.25, 0.5) {};
		\node [style=none] (6) at (-1.25, 0.5) {};
		\node [style=none, rotate=90] (7) at (-1, -0.25) {...};
		\node [style=none, rotate=90] (8) at (1, -0.25) {...};
	\end{pgfonlayer}
	\begin{pgfonlayer}{edgelayer}
		\draw [in=-141, out=0, looseness=0.75] (3.center) to (0);
		\draw [in=180, out=-39, looseness=0.75] (0) to (4.center);
		\draw [in=180, out=22, looseness=0.75] (0) to (5.center);
		\draw [in=180, out=39, looseness=0.75] (0) to (1.center);
		\draw [in=0, out=158, looseness=0.75] (0) to (6.center);
		\draw [in=141, out=0, looseness=0.75] (2.center) to (0);
	\end{pgfonlayer}
\end{tikzpicture} \ &:= \ \ketbra{\texttt{0}\cdots\texttt{0}}{\texttt{0}\cdots\texttt{0}} \,+\,
\mathrm e^{i \alpha} \,\ketbra{\texttt{1}\cdots\texttt{1}}{\texttt{1}\cdots\texttt{1}}
\\[2ex]
\begin{tikzpicture}
	\begin{pgfonlayer}{nodelayer}
		\node [style=X phase dot] (0) at (0, 0) {$\alpha$};
		\node [style=none] (1) at (1.25, 1) {};
		\node [style=none] (2) at (-1.25, 1) {};
		\node [style=none] (3) at (-1.25, -1) {};
		\node [style=none] (4) at (1.25, -1) {};
		\node [style=none] (5) at (1.25, 0.5) {};
		\node [style=none] (6) at (-1.25, 0.5) {};
		\node [style=none, rotate=90] (7) at (-1, -0.25) {...};
		\node [style=none, rotate=90] (8) at (1, -0.25) {...};
	\end{pgfonlayer}
	\begin{pgfonlayer}{edgelayer}
		\draw [in=-141, out=0, looseness=0.75] (3.center) to (0);
		\draw [in=180, out=-39, looseness=0.75] (0) to (4.center);
		\draw [in=180, out=22, looseness=0.75] (0) to (5.center);
		\draw [in=180, out=39, looseness=0.75] (0) to (1.center);
		\draw [in=0, out=158, looseness=0.75] (0) to (6.center);
		\draw [in=141, out=0, looseness=0.75] (2.center) to (0);
	\end{pgfonlayer}
\end{tikzpicture} \ &:= \ \ketbra{\texttt{+}\cdots\texttt{+}}{\texttt{+}\cdots\texttt{+}} \,+\,
\mathrm e^{i \alpha} \,\ketbra{\texttt{-}\cdots\texttt{-}}{\texttt{-}\cdots\texttt{-}}
\end{aligned}
\end{equation}
Note that if you are reading this document in monochrome or otherwise have difficulty distinguishing green and red, $Z$ spiders will appear lightly-shaded and $X$ darkly-shaded.
The diagram as a whole corresponds to a linear map built from the
spiders (and permutations) by the usual composition and tensor product
of linear maps.  As a special case, diagrams with no inputs represent
(super-normalised) state preparations.
Note that when $\alpha=0$, we will not write the phase on the spider.

\begin{example}\label{ex:basic-maps-and-states}
  We can immediately write down some simple state preparations and
  unitaries in the \zxcalculus:
  \begin{equation}
	\label{eq:basic-maps-and-states}
  \begin{aligned}
  \begin{array}{rcccl}
  \begin{tikzpicture}
	\begin{pgfonlayer}{nodelayer}
		\node [style=Z dot] (0) at (-0.5, 0) {};
		\node [style=none] (1) at (0.5, 0) {};
	\end{pgfonlayer}
	\begin{pgfonlayer}{edgelayer}
		\draw (0) to (1.center);
	\end{pgfonlayer}
\end{tikzpicture} & = & \ket{\texttt{0}} + \ket{\texttt{1}}& \ =& \sqrt{2}\, \ket{\texttt{+}} \\[2ex]
  \begin{tikzpicture}
	\begin{pgfonlayer}{nodelayer}
		\node [style=X dot] (0) at (-0.5, 0) {};
		\node [style=none] (1) at (0.5, 0) {};
	\end{pgfonlayer}
	\begin{pgfonlayer}{edgelayer}
		\draw (0) to (1.center);
	\end{pgfonlayer}
\end{tikzpicture} & = & \ket{\texttt{+}} + \ket{\texttt{-}}& \ =& \sqrt{2}\, \ket{\texttt{0}} \\[2ex]
  \begin{tikzpicture}
	\begin{pgfonlayer}{nodelayer}
		\node [style=Z phase dot] (0) at (0, 0) {$\alpha$};
		\node [style=none] (1) at (-1.25, 0) {};
		\node [style=none] (2) at (1.25, 0) {};
	\end{pgfonlayer}
	\begin{pgfonlayer}{edgelayer}
		\draw (1.center) to (0);
		\draw (0) to (2.center);
	\end{pgfonlayer}
\end{tikzpicture} & = & \ketbra{\texttt{0}}{\texttt{0}} \,+\, \mathrm e^{i \alpha} \,\ketbra{\texttt{1}}{\texttt{1}}&\ =:\, & Z_\alpha \\[2ex]
  \begin{tikzpicture}
	\begin{pgfonlayer}{nodelayer}
		\node [style=X phase dot] (0) at (0, 0) {$\alpha$};
		\node [style=none] (1) at (-1.25, 0) {};
		\node [style=none] (2) at (1.25, 0) {};
	\end{pgfonlayer}
	\begin{pgfonlayer}{edgelayer}
		\draw (1.center) to (0);
		\draw (0) to (2.center);
	\end{pgfonlayer}
\end{tikzpicture} & = & \ketbra{\texttt{+}}{\texttt{+}} \,+\, \mathrm e^{i \alpha} \,\ketbra{\texttt{-}}{\texttt{-}}&\ =:\, & X_\alpha
  \end{array}  	
  \end{aligned}
  \end{equation}
  We can also represent the effects that are dual to the states above using spiders:
  \begin{equation}
  \begin{aligned}
  \begin{array}{rcccl}
  \begin{tikzpicture}
	\begin{pgfonlayer}{nodelayer}
		\node [style=Z dot] (0) at (0.5, 0) {};
		\node [style=none] (1) at (-0.5, 0) {};
	\end{pgfonlayer}
	\begin{pgfonlayer}{edgelayer}
		\draw (0) to (1.center);
	\end{pgfonlayer}
\end{tikzpicture}
 & = & \bra{\texttt{0}} + \bra{\texttt{1}}& \ =& \sqrt{2}\, \bra{\texttt{+}} \\[2ex]
  \begin{tikzpicture}
	\begin{pgfonlayer}{nodelayer}
		\node [style=X dot] (0) at (0.5, 0) {};
		\node [style=none] (1) at (-0.5, 0) {};
	\end{pgfonlayer}
	\begin{pgfonlayer}{edgelayer}
		\draw (0) to (1.center);
	\end{pgfonlayer}
\end{tikzpicture}
 & = & \bra{\texttt{+}} + \bra{\texttt{-}}& \ =& \sqrt{2}\, \bra{\texttt{0}} \\[2ex]
  \end{array}   
  \end{aligned}
  \end{equation}
  
\end{example}
In the diagrams above we write explicit scalars to represent a proportionality constant.
In this paper (non-zero) scalar factors will not be important. However, do note it is always possible to represent any scalar as an explicit ZX-diagram (of constant size).
For this reason, our results will also apply to other proposed normalisations of the ZX generators, such as those in Refs.~\cite{debeaudrap2020welltempered,east2020akltstates,d.p.east2021spinnetworks}.

It is often convenient to introduce a symbol --- a yellow square --- for
the Hadamard gate.
This is defined by the equation:
\begin{equation}
\label{eq:Hdef}
\qquad
\begin{tikzpicture}
	\begin{pgfonlayer}{nodelayer}
		\node [style=none] (0) at (8, 0) {};
		\node [style=none] (1) at (10, 0) {};
		\node [style=hadamard] (2) at (9, 0) {};
		\node [style=none] (3) at (7, 0) {$=:$};
		\node [style=none] (4) at (2.25, 0) {};
		\node [style=Z phase dot] (5) at (3, 0) {$\frac\pi2$};
		\node [style=X phase dot] (6) at (4, 0) {$\frac\pi2$};
		\node [style=Z phase dot] (7) at (5, 0) {$\frac\pi2$};
		\node [style=none] (8) at (5.75, 0) {};
		\node [style=box] (9) at (-3.25, 0) {$H$};
		\node [style=none] (10) at (-5, 0) {};
		\node [style=none] (11) at (-1.75, 0) {};
		\node [style=none] (12) at (-0.75, 0) {=};
		\node [style=none] (13) at (1.25, 0.075) {$\mathrm e^{-i\pi/4}\;$};
	\end{pgfonlayer}
	\begin{pgfonlayer}{edgelayer}
		\draw (0.center) to (2);
		\draw (2) to (1.center);
		\draw (4.center) to (8.center);
		\draw (10.center) to (11.center);
	\end{pgfonlayer}
\end{tikzpicture}
\end{equation}
The \CNOT gate also has a straightforward representation as a ZX-diagram:
\begin{equation}
  \CNOT \ =\ \sqrt{2}\ \ \begin{tikzpicture}
	\begin{pgfonlayer}{nodelayer}
		\node [style=Z dot] (0) at (0, 0.5) {};
		\node [style=none] (1) at (1.25, 0.5) {};
		\node [style=X dot] (2) at (0, -0.5) {};
		\node [style=none] (3) at (-1.25, 0.5) {};
		\node [style=none] (4) at (1.25, -0.5) {};
		\node [style=none] (5) at (-1.25, -0.5) {};
	\end{pgfonlayer}
	\begin{pgfonlayer}{edgelayer}
		\draw (0) to (2);
		\draw (0) to (3.center);
		\draw (0) to (1.center);
		\draw (5.center) to (2);
		\draw (2) to (4.center);
	\end{pgfonlayer}
\end{tikzpicture}
\end{equation}
Seeing as can represent $Z_\alpha$, $H$ and $\CNOT$ gates as ZX-diagrams, we see that we can in fact represent any unitary as a ZX-diagram.
The above demonstrates that ZX-diagrams can be used as an alternative representation for quantum circuits.
However, ZX-diagrams are also more versatile than unitary circuits. Consider for example the following construction of the \CZ gate as a ZX-diagram:
\begin{align}
\label{eq:zx-derived-gates}
  \CZ\ &\propto\ \begin{tikzpicture}
	\begin{pgfonlayer}{nodelayer}
		\node [style=Z dot] (0) at (-2.75, 0.5) {};
		\node [style=none] (1) at (-1.5, 0.5) {};
		\node [style=X dot] (2) at (-2.75, -0.5) {};
		\node [style=none] (3) at (-4, 0.5) {};
		\node [style=none] (4) at (-1.5, -0.5) {};
		\node [style=none] (5) at (-4, -0.5) {};
		\node [style=hadamard] (6) at (-3.5, -0.5) {};
		\node [style=hadamard] (7) at (-2, -0.5) {};
		\node [style=none] (8) at (-0.75, 0) {$=$};
		\node [style=Z dot] (9) at (1, 0.5) {};
		\node [style=none] (10) at (2, 0.5) {};
		\node [style=Z dot] (11) at (1, -0.5) {};
		\node [style=none] (12) at (0, 0.5) {};
		\node [style=none] (13) at (2, -0.5) {};
		\node [style=none] (14) at (0, -0.5) {};
		\node [style=hadamard] (15) at (1, 0) {};
	\end{pgfonlayer}
	\begin{pgfonlayer}{edgelayer}
		\draw (0) to (2);
		\draw (0) to (3.center);
		\draw (0) to (1.center);
		\draw (5.center) to (2);
		\draw (2) to (4.center);
		\draw (9) to (11);
		\draw (9) to (12.center);
		\draw (9) to (10.center);
		\draw (14.center) to (11);
		\draw (11) to (13.center);
	\end{pgfonlayer}
\end{tikzpicture} \;.
\end{align}
The right-hand-side demonstrates a different diagrammatic construction for \CZ, that does not immediately look circuit-like, with the Hadamard-box representing some sort of interaction of two qubits rather than the evolution of a single qubit.

In fact, this versatility is reflected in the property that ZX-diagrams are universal for \emph{all} linear maps between any number of qubits. To see this, note that we can represent states as in Eq.~\eqref{eq:basic-maps-and-states}. By composing tensor products of these states with some unitary we can write down any quantum state. By the map-state duality of quantum theory (i.e.~the Choi-Jamio\l{}kowski isomorphism), we can then also write every linear map, see for instance~\cite{vandewetering2020zxcalculus} for the details.

The universality of the gate-set $\mathcal B$ and of the ZX-calculus means that any unitary operator on some fixed number of qubits may be represented by some `gadget' in the ZX calculus, consisting of some fixed diagram of finite size --- though as the example of CZ in Eq.~\eqref{eq:zx-derived-gates} shows, there may also be `gadgets' which represent a unitary operator which do \emph{not} consist of sequential and parallel composition of diagrams from Eq.~\eqref{eqn:spiders}. Indeed, even the representation of the \CNOT\ is by a simple `gadget' of two nodes, which is not describable as a composition of the other single-node `gadgets'.
In this respect, \zxdiagrams\ represent a more versatile notation than a conventional circuit notation.
This raises the question of how, given a representation of some unitary $U$ as a \zxdiagram, one might find another representation of $U$ which consists of just compositions from the universal gate-set $\mathcal B$.
This is the problem that this paper is concerned with.

\zxdiagrams\ are more than just a notation for unitary circuits (and non-unitary operators more generally): they may be used to perform computations.
Specifically, \zxdiagrams\ come with a set of graphical rewrite rules, which may be used to find equivalent diagrams which represent the same state or operator, just as one might manipulate an algebraic expression.
This rewrite system is \emph{complete}~\cite{HarnyCompleteness,vilmarteulercompleteness}: unlike other circuit diagrams, one may show that two equivalent \zxdiagrams\ are equivalent though transformations of diagrams alone.
The possible advantage of this is that \zxdiagrams\ can often concisely represent operators which have a very large number of non-zero coefficients, and so that this reasoning can be done efficiently while it could not be done using the matrices directly.
For instance, one of the rewrite rules we will use in this paper is \emph{spider fusion}:
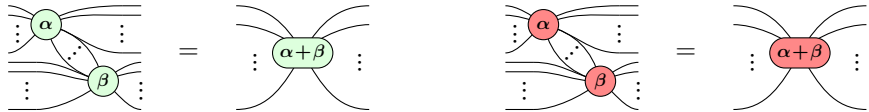
\[\begin{tikzpicture}
	\begin{pgfonlayer}{nodelayer}
		\node [style=Z phase dot] (0) at (-2.25, -0.75) {$\beta$};
		\node [style=none] (1) at (-1.25, -0.25) {};
		\node [style=none] (2) at (-4, -0.25) {};
		\node [style=none] (3) at (-4, -1.5) {};
		\node [style=none] (4) at (-1.25, -1.5) {};
		\node [style=none] (5) at (-1.25, -0.5) {};
		\node [style=none] (6) at (-4, -0.5) {};
		\node [style=none, rotate=90] (7) at (-4.25, -1) {...};
		\node [style=none, rotate=90] (8) at (-1.25, -1) {...};
		\node [style=none] (9) at (-4.75, 1) {};
		\node [style=Z phase dot] (10) at (-3.75, 0.75) {$\alpha$};
		\node [style=none] (11) at (-1.75, 1.25) {};
		\node [style=none] (12) at (-4.75, 1.25) {};
		\node [style=none] (13) at (-1.75, 0) {};
		\node [style=none, rotate=90] (14) at (-1.75, 0.5) {...};
		\node [style=none] (15) at (-1.75, 1) {};
		\node [style=none] (16) at (-4.75, 0) {};
		\node [style=none, rotate=90] (17) at (-4.5, 0.5) {...};
		\node [style=none] (18) at (-1.25, 0) {};
		\node [style=none] (19) at (-1.25, 1.25) {};
		\node [style=none] (20) at (-1.25, 1) {};
		\node [style=none] (21) at (-4.75, -0.5) {};
		\node [style=none] (22) at (-4.75, -1.5) {};
		\node [style=none] (23) at (-4.75, -0.25) {};
		\node [style=none] (24) at (0, 0) {$=$};
		\node [style=none, rotate=45] (25) at (-3, 0) {...};
		\node [style=none] (26) at (1.25, 1.25) {};
		\node [style=none] (27) at (4.75, 0.75) {};
		\node [style=none, rotate=90] (28) at (1.75, -0.25) {...};
		\node [style=none] (29) at (4.75, -1.5) {};
		\node [style=none, rotate=90] (30) at (4.5, -0.25) {...};
		\node [style=none] (31) at (1.25, 0.75) {};
		\node [style=Z phase dot] (32) at (3, 0) {$\ \alpha\!+\!\beta\ $};
		\node [style=none] (33) at (4.75, 1.25) {};
		\node [style=none] (34) at (1.25, -1.5) {};
	\end{pgfonlayer}
	\begin{pgfonlayer}{edgelayer}
		\draw [style=simple, in=-141, out=0, looseness=0.75] (3.center) to (0);
		\draw [style=simple, in=180, out=-39, looseness=0.75] (0) to (4.center);
		\draw [style=simple, in=180, out=22, looseness=0.75] (0) to (5.center);
		\draw [style=simple, in=180, out=39, looseness=0.75] (0) to (1.center);
		\draw [style=simple, in=0, out=180] (0) to (6.center);
		\draw [style=simple, in=165, out=0] (2.center) to (0);
		\draw [style=simple, in=-141, out=0, looseness=0.75] (16.center) to (10);
		\draw [style=simple, in=180, out=-15] (10) to (13.center);
		\draw [style=simple, in=180, out=22] (10) to (15.center);
		\draw [style=simple, in=180, out=39, looseness=0.75] (10) to (11.center);
		\draw [style=simple, in=0, out=158, looseness=0.75] (10) to (9.center);
		\draw [style=simple, in=141, out=0, looseness=0.75] (12.center) to (10);
		\draw [style=simple, bend left] (10) to (0);
		\draw [style=simple] (11.center) to (19.center);
		\draw [style=simple] (15.center) to (20.center);
		\draw [style=simple] (13.center) to (18.center);
		\draw [style=simple] (23.center) to (2.center);
		\draw [style=simple] (21.center) to (6.center);
		\draw [style=simple] (22.center) to (3.center);
		\draw [style=simple, bend left] (0) to (10);
		\draw [style=simple, in=-120, out=0] (34.center) to (32);
		\draw [style=simple, in=180, out=-60] (32) to (29.center);
		\draw [style=simple, in=180, out=45, looseness=0.75] (32) to (27.center);
		\draw [style=simple, in=180, out=60] (32) to (33.center);
		\draw [style=simple, in=0, out=135, looseness=0.75] (32) to (31.center);
		\draw [style=simple, in=120, out=0] (26.center) to (32);
	\end{pgfonlayer}
\end{tikzpicture}\qquad\qquad\quad \begin{tikzpicture}
	\begin{pgfonlayer}{nodelayer}
		\node [style=X phase dot] (0) at (-2.25, -0.75) {$\beta$};
		\node [style=none] (1) at (-1.25, -0.25) {};
		\node [style=none] (2) at (-4, -0.25) {};
		\node [style=none] (3) at (-4, -1.5) {};
		\node [style=none] (4) at (-1.25, -1.5) {};
		\node [style=none] (5) at (-1.25, -0.5) {};
		\node [style=none] (6) at (-4, -0.5) {};
		\node [style=none, rotate=90] (7) at (-4.25, -1) {...};
		\node [style=none, rotate=90] (8) at (-1.25, -1) {...};
		\node [style=none] (9) at (-4.75, 1) {};
		\node [style=X phase dot] (10) at (-3.75, 0.75) {$\alpha$};
		\node [style=none] (11) at (-1.75, 1.25) {};
		\node [style=none] (12) at (-4.75, 1.25) {};
		\node [style=none] (13) at (-1.75, 0) {};
		\node [style=none, rotate=90] (14) at (-1.75, 0.5) {...};
		\node [style=none] (15) at (-1.75, 1) {};
		\node [style=none] (16) at (-4.75, 0) {};
		\node [style=none, rotate=90] (17) at (-4.5, 0.5) {...};
		\node [style=none] (18) at (-1.25, 0) {};
		\node [style=none] (19) at (-1.25, 1.25) {};
		\node [style=none] (20) at (-1.25, 1) {};
		\node [style=none] (21) at (-4.75, -0.5) {};
		\node [style=none] (22) at (-4.75, -1.5) {};
		\node [style=none] (23) at (-4.75, -0.25) {};
		\node [style=none] (24) at (0, 0) {$=$};
		\node [style=none, rotate=45] (25) at (-3, 0) {...};
		\node [style=none] (26) at (1.25, 1.25) {};
		\node [style=none] (27) at (4.75, 0.75) {};
		\node [style=none, rotate=90] (28) at (1.75, -0.25) {...};
		\node [style=none] (29) at (4.75, -1.5) {};
		\node [style=none, rotate=90] (30) at (4.5, -0.25) {...};
		\node [style=none] (31) at (1.25, 0.75) {};
		\node [style=X phase dot] (32) at (3, 0) {$\ \alpha\!+\!\beta\ $};
		\node [style=none] (33) at (4.75, 1.25) {};
		\node [style=none] (34) at (1.25, -1.5) {};
	\end{pgfonlayer}
	\begin{pgfonlayer}{edgelayer}
		\draw [style=simple, in=-141, out=0, looseness=0.75] (3.center) to (0);
		\draw [style=simple, in=180, out=-39, looseness=0.75] (0) to (4.center);
		\draw [style=simple, in=180, out=22, looseness=0.75] (0) to (5.center);
		\draw [style=simple, in=180, out=39, looseness=0.75] (0) to (1.center);
		\draw [style=simple, in=0, out=180] (0) to (6.center);
		\draw [style=simple, in=165, out=0] (2.center) to (0);
		\draw [style=simple, in=-141, out=0, looseness=0.75] (16.center) to (10);
		\draw [style=simple, in=180, out=-15] (10) to (13.center);
		\draw [style=simple, in=180, out=22] (10) to (15.center);
		\draw [style=simple, in=180, out=39, looseness=0.75] (10) to (11.center);
		\draw [style=simple, in=0, out=158, looseness=0.75] (10) to (9.center);
		\draw [style=simple, in=141, out=0, looseness=0.75] (12.center) to (10);
		\draw [style=simple, bend left] (10) to (0);
		\draw [style=simple] (11.center) to (19.center);
		\draw [style=simple] (15.center) to (20.center);
		\draw [style=simple] (13.center) to (18.center);
		\draw [style=simple] (23.center) to (2.center);
		\draw [style=simple] (21.center) to (6.center);
		\draw [style=simple] (22.center) to (3.center);
		\draw [style=simple, bend left] (0) to (10);
		\draw [style=simple, in=-120, out=0] (34.center) to (32);
		\draw [style=simple, in=180, out=-60] (32) to (29.center);
		\draw [style=simple, in=180, out=45, looseness=0.75] (32) to (27.center);
		\draw [style=simple, in=180, out=60] (32) to (33.center);
		\draw [style=simple, in=0, out=135, looseness=0.75] (32) to (31.center);
		\draw [style=simple, in=120, out=0] (26.center) to (32);
	\end{pgfonlayer}
\end{tikzpicture}\]
These rules say that we can fuse together adjacent spiders of the same colour.

While these rewrite rules are not immediately relevant to our results, the fact that it is possible to compute with \zxdiagrams\ is the motivation for considering this particular representation of unitary circuits, and also motivates the concept of considering different \zxdiagrams\ which represent the same unitary transformation.
We refer the interested reader to~\cite{vandewetering2020zxcalculus} for an overview.

\subsection{Circuit extraction}

In the above section we saw that we can get ZX-diagrams directly from quantum circuits.
We can also get ZX-diagrams from considering measurement patterns in the \emph{one-way model}~\cite{MBQC1}. In the one-way model of quantum computation we start with a large \emph{graph state}, on which we then do subsequent measurements, where the choice of measurement angle and axis may depend on previous measurement outcomes. This leads to another universal model of quantum computation. The one-way model can be straightforwardly represented in the ZX-calculus~\cite{DP2,Backens2020extraction}.

An important property of a one-way computation is that we can perform a computation deterministically, so that we perform the same overall computation regardless of individual measurement outcomes. A sufficient property for ensuring that deterministic processes are possible on a given resource state is that its underlying graph has a property known as \emph{gflow}~\cite{GFlow}.
This is an efficiently verifiable combinatorial condition on the entangled resource. 

When we represent a one-way computation with gflow as a ZX-diagram, the gflow ensures that certain `local' parts of the diagram correspond to individual unitary gates, in a way which can be iteratively translated into an actual unitary circuit.
In this case we can hence \emph{extract} a unitary quantum circuit from the ZX-diagram that represents the one-way computation. See for instance~\cite{cliffsimp,Backens2020extraction,Simmons2021Measurement} for several variations on this idea.

Measurement-based quantum computation like the one-way model is a type of non-unitary quantum computation. Another type of non-unitary model is given by doing \emph{lattice surgery} in the surface code~\cite{horsman2011quantum,horsman2017surgery}. A lattice surgery procedure can also be represented as a ZX-diagram~\cite{horsman2017surgery}. Just as in the one-way model, there is a flow condition that ensures such a calculation is deterministic, and that the resulting ZX-diagram can be step-by-step rewritten into a unitary circuit~\cite{deBeaudrap2020Paulifusion}.

We see that there are several quantum computational models that can be written in terms of ZX-diagrams, which can be rewritten into a unitary quantum circuit efficiently when they satisfy some condition. The type of flow condition required for these procedures ensures that the diagram can't get `too wild' in the middle, so that we can stepwise rewrite the diagram into something that looks more like a circuit.
A natural question to ask then is how much we can weaken such additional conditions, and in particular if we can transform a ZX-diagram into a circuit efficiently in the most general setting, where the only condition we require of the ZX-diagram is that it is proportional to a unitary.
The main result of this paper is that such a general efficient procedure most likely does not exist.

\subsection{Background on computational complexity}

Finally, we provide some background on computational complexity.
We assume knowledge of~\P, the boolean satisfiability problem~$\mathbf{SAT}$, oracle machines,~\NP and nondeterministic Turing machines (NTMs) in general.
Our results concern \emph{Cook reductions} (in fact, usually Cook[1] reductions). A Cook reduction from a problem $\mathbf X$ to another problem $\mathbf Z$ is an algorithm for solving $\mathbf X$ using a deterministic Turing machine which halts in polynomial time, but which may query an oracle (in the case of a Cook[1] reduction, exactly once) for $\mathbf Z$.
This implies that, modulo some polynomial-time computation, the problem $\mathbf Z$ is at least as hard as $\mathbf X$; and that if $\mathbf Z \in \P$, we also have $\mathbf X \in \P$. In symbols we may write $\mathbf X \in \P^{\mathbf Z}$.
Our results will generally concern problems $\mathbf Z$ related to ZX-diagrams and problems $\mathbf X$ which are at least \NP-hard (\emph{i.e.},~they suffice to solve $\mathbf{SAT}$).

Quantum circuits (specifically: uniform circuit families, as described in Appendix~\ref{apx:circuits}) form a model of computation, which may be considered to generate random outcomes through measurement operations.
The class \BQP\ consists of decision problems which can be decided with bounded error (with error probability less than $\tfrac{1}{3}$, say) by such circuit families, and represents the decision problems that can be practically solved by an (idealised) quantum computer.
It is not expected that either of $\NP$ or $\BQP$ contain the other. So if we can reduce in polynomial time (by many-to-one or oracle reductions) an \NP-complete problem to some problem $\mathbf X$, then we expect $\mathbf X$ to be intractable for quantum computers.
Certain modifications of the quantum computational model do allow for more difficult problems to be solved, however.
For instance, \PostBQP\ is the class of problems which may be solved with bounded error by a uniform quantum circuit family, \emph{conditioned on some other measurement} yielding a specific outcome (which occurs with non-zero probability).
This `conditioning' restriction is known as \emph{postselection}, and appears to be operationally very powerful, as \PostBQP coincides with the class \PP, of decision problems for which a `yes' instance is accepted on more than half of the branches of some NTM halting in polynomial time.

The class \P\ is closed under oracles: a deterministic Turing machine equipped with an oracle for some problem in \P cannot decide more problems in polynomial time than a normal Turing machine, so that $\P^\P = \P$.
The same is true for \BQP: any decision problem solvable (with bounded error) by a uniform family of quantum circuits, can also be solved (with bounded error) by some other family of quantum circuits without oracle access, so that $\BQP^\BQP = \BQP$.
The same is not true, however, for \NP: it is not known whether $\NP^\NP$ (the class of decision problems, for which there is an NTM with an oracle for a problem in \NP, halts in polynomial time and accepts in some branch precisely for `yes' instances) is equal to \NP.
It is widely conjectured that $\Sigma_2^{\mathrm p} := \NP^\NP \ne \NP$, and indeed that $\Sigma_3^{\mathrm p} := \NP^{\Sigma_2^{\mathrm p}} = \NP^{\NP^\NP} \ne \NP^\NP$, and so forth.
The union of $\Sigma_n^{\mathrm p} := \NP^{\Sigma_{n-1}^{\mathrm p}}$ for all $n>1$, defines the class \PH, called the \emph{polynomial hierarchy}~\cite{Stockmeyer-1977}.

The hardness results which we are most concerned with involve problems in \sharpP: the class of problems which may be reduced to counting the number of accepting branches of some NTM on a given input.
In particular, we are interested in the problem \numberSAT, of counting the number of `solutions' $x \in \{0,1\}^n$ to an instance of \textbf{SAT}, presented as a formula for a function $f: \{0,1\}^n \to \{0,1\}$, where a `solution' satisfies $f(x) = 1$.
The problem \numberSAT\ is \sharpP-complete~\cite{Valiant-1979}, as is tensor contraction over the natural numbers~\cite{tensors-hard}, and `strong simulation' (\emph{i.e.},~precise estimation of explicit measurement probabilities) of uniform quantum circuit families~\cite{nest2010clifford}.
The \sharpP-completeness here means that a Cook reduction from any of these problems to some problem $\mathbf Z$, establishes that there is a Cook reduction from \emph{any} problem $\mathbf X \in \sharpP$ to $\mathbf Z$.
in this case we say then that $\mathbf Z$ is ``\sharpP-hard''.
The computational power of \sharpP\ is considered to be significantly greater than that of \NP. 
In particular, Toda~\cite{Toda-1991} showed that $\PH \subseteq \P^\sharpP$.

\section{Proof of hardness of Circuit extraction}\label{sec:circuitExtraction}

We now present the central problem of our work.

\begin{quote}
\noindent $\CE$ \\
\noindent \textbf{Input}: A ZX-diagram $D$ with $n$ inputs and outputs and at most $k$ wires and/or spiders, and a set $\mathcal G$ of unitary gates (each acting on at most $O(1)$ qubits). \\
\noindent \textbf{Promise}: The operator denoted by $D$ is proportional to a unitary. \\
\noindent \textbf{Output}: Either \textbf{(a)}~a~$\text{poly}(n,k)$-size circuit $C$, expressed as a sequence of gates from $\mathcal G$  and expressing an $n$-qubit unitary that is proportional to the operator denoted by $D$, if such a circuit exists; or \textbf{(b)}~a~message that no such circuit exists, if that is the case.
\end{quote}

Note that here we make no assumptions on the specific gate set $\mathcal G$, apart from the computability of the coefficients as described in Section~\ref{sec:circuits}, and that the number of qubits which is bounded by some constant. One might object to the requirement that the output list of gates must be polynomially related to the size of the input ZX-diagram: however, as we are interested in whether the extraction problem can be solved efficiently, the restriction on the size of $C$ follows from the time required to represent it as a list of gates.

The above problem can of course also be stated for any related graphical language for quantum operations, such as the ZH-calculus~\cite{backens2018zhcalculus} or the ZW-calculus~\cite{hadzihasanovic2015diagrammatic}. Since such diagrams can be efficiently translated into one another, these problems are of equivalent hardness.
There are some other reasonable variations we can consider of \textbf{CircuitExtraction} that we will discuss in the next section.

We will now show that $\CE$ is \sharpP-hard. We do this by building a diagram that is proportional to a unitary based on a SAT instance, and showing that the resulting matrix the diagram represents is uniquely determined by the number of solutions of the instance.

Let $f:\{0,1\}^n\to \{0,1\}$ be a Boolean formula with $\text{poly}(n)$ terms. We say a bit string $x\in\{0,1\}^n$ is a solution to $f$ when $f(x) = 1$.
The first step will be to build a ZX-diagram that implements the linear map $L_f$ that takes $n$ qubits to $1$ qubit by $L_f\ket{x} = \ket{f(x)}$.
We can of course represent $f$ as a tree of AND and NOT operations so that to construct $L_f$ it suffices to find linear maps that implement AND and NOT on $\ket{x}$.

We may consider ZX diagrams for ``quantum'' versions of the boolean logical AND gate and NOT gate, \emph{i.e.},~linear operators such that $\text{NOT}\ket{\texttt 0} = \ket{\texttt  1}$, $\text{NOT}\ket{\texttt 1} = \ket{\texttt 0}$, and $\text{AND}\ket{x,y} \mapsto \ket{x\cdot y}$.
These operations can be represented (up to a constant factor) by the following ZX-diagrams:
\[\text{NOT} \ \ = \ \ \begin{tikzpicture}
	\begin{pgfonlayer}{nodelayer}
		\node [style=X phase dot] (0) at (0, 0) {$\pi$};
		\node [style=none] (1) at (-1.25, 0) {};
		\node [style=none] (2) at (1.25, 0) {};
	\end{pgfonlayer}
	\begin{pgfonlayer}{edgelayer}
		\draw (1.center) to (0);
		\draw (0) to (2.center);
	\end{pgfonlayer}
\end{tikzpicture}\qquad\qquad \text{AND} \ \  \propto \ \ \begin{tikzpicture}
	\begin{pgfonlayer}{nodelayer}
		\node [style=Z dot] (0) at (-1.5, 1) {};
		\node [style=Z dot] (1) at (-1.5, -1) {};
		\node [style=Z phase dot] (2) at (0, 0) {-$\frac\pi4$};
		\node [style=X phase dot] (3) at (2, 0) {-$\frac\pi2$};
		\node [style=none] (4) at (-3.5, 1) {};
		\node [style=none] (5) at (-3.5, -1) {};
		\node [style=none] (9) at (3.5, 0) {};
		\node [style=X dot] (10) at (0, -1) {};
		\node [style=X dot] (11) at (0, 1) {};
		\node [style=X dot] (12) at (-1.5, 0) {};
		\node [style=Z phase dot] (13) at (1.25, 1) {-$\frac\pi4$};
		\node [style=Z phase dot] (14) at (1.25, -1) {-$\frac\pi4$};
		\node [style=Z phase dot] (15) at (-2.5, 0) {$\frac\pi4$};
	\end{pgfonlayer}
	\begin{pgfonlayer}{edgelayer}
		\draw (2) to (3);
		\draw (4.center) to (0);
		\draw (5.center) to (1);
		\draw (3) to (9.center);
		\draw (10) to (14);
		\draw (11) to (13);
		\draw (12) to (15);
		\draw (0) to (12);
		\draw (1) to (12);
		\draw (2) to (12);
		\draw (0) to (11);
		\draw (2) to (11);
		\draw (1) to (10);
		\draw (2) to (10);
	\end{pgfonlayer}
\end{tikzpicture}\]
The NOT gate is just an $X_\pi$ gate, but the AND is more complicated. It is based on the 4 $T$-gate representation of the CCZ gate from~\cite{kissinger2021simulating}.

By combining these diagrammatic gadgets for NOT and AND we can build the operation $L_f$ as a ZX-diagram using $\text{poly}(n)$ spiders.
Now, note that:
\begin{equation}\label{eq:boolean-output-zx}
  \begin{tikzpicture}
	\begin{pgfonlayer}{nodelayer}
		\node [style=Z dot] (12) at (-1.25, 0.75) {};
		\node [style=none] (13) at (1.25, 0) {};
		\node [style=box, minimum height=1cm] (14) at (0, 0) {$L_f$};
		\node [style=none] (23) at (-0.25, 0.75) {};
		\node [style=Z dot] (24) at (-1.25, -0.75) {};
		\node [style=none] (25) at (-0.25, -0.75) {};
		\node [style=none] (26) at (-1.25, 0.25) {$\vdots$};
	\end{pgfonlayer}
	\begin{pgfonlayer}{edgelayer}
		\draw (12) to (23.center);
		\draw (24) to (25.center);
		\draw (14) to (13.center);
	\end{pgfonlayer}
\end{tikzpicture} \ = \ \sum_{x} L_x \ket{x} \ =\  \sum_{x} \ket{f(x)} \ = \ \frac{N_0}{2^n} \ket{\texttt0} + \frac{N_1}{2^n} \ket{\texttt1} \ =:\ a_0 \ket{\texttt0} + a_1\ket{\texttt1}
\end{equation}
where $N_1$ is the number of solutions of $f$, $N_0 = 2^n - N_1$ is the number of `non-solutions' of $f$, and we set $a_0 = N_0/N$ and $a_1 = N_1/N$ for $N := 2^n = N_0 + N_1$.
The resulting state is not normalised: to normalise it we should multiply both sides by $(a_0^2+a_1^2)^{-1/2}$.

We use the `state' described in Eq.~~\eqref{eq:boolean-output-zx} as the input of a controlled operation.
By choosing the controlled operation appropriately, we will be left with something proportional to a unitary.
We may for instance consider the following diagram:
\begin{equation}\label{eq:unitary-sat-zx}
  \begin{tikzpicture}
	\begin{pgfonlayer}{nodelayer}
		\node [style=Z dot] (12) at (-4, 2) {};
		\node [style=none] (13) at (-1.5, 1.25) {};
		\node [style=box, minimum height=1cm] (14) at (-2.75, 1.25) {$L_f$};
		\node [style=none] (23) at (-3, 2) {};
		\node [style=Z dot] (24) at (-4, 0.5) {};
		\node [style=none] (25) at (-3, 0.5) {};
		\node [style=none] (26) at (-4, 1.5) {$\vdots$};
		\node [style=none] (27) at (-2, -1) {};
		\node [style=none] (28) at (1.25, -1) {};
		\node [style=X dot] (29) at (-0.25, -1) {};
		\node [style=Z phase dot] (30) at (-0.25, 0.25) {$\ -\frac\pi2~$};
	\end{pgfonlayer}
	\begin{pgfonlayer}{edgelayer}
		\draw (12) to (23.center);
		\draw (24) to (25.center);
		\draw (14) to (13.center);
		\draw (27.center) to (28.center);
		\draw [in=90, out=0, looseness=1.25] (13.center) to (30);
		\draw (30) to (29);
	\end{pgfonlayer}
\end{tikzpicture}
\end{equation}
To see this is unitary first recall that a $Y$ rotation over an angle $\alpha$ applied to $\ket{\texttt 0}$ gives $Y_\alpha\ket{\texttt 0} = \cos(\frac\alpha2)\ket{\texttt 0} + \sin(\frac\alpha2)\ket{\texttt 1}$. Hence the state of Eq.~\eqref{eq:boolean-output-zx}, when properly normalised, can be written as $Y_\alpha\ket{\texttt 0}$ for $\alpha = 2\sin^{-1}\left(\smash{\frac{a_1}{\sqrt{a_0^2+a_1^2}}}\right)$.
We can then calculate:
\begin{equation}\label{eq:unitary-sat-zx2}
  \begin{tikzpicture}
	\begin{pgfonlayer}{nodelayer}
		\node [style=Z dot] (12) at (-12.25, 1) {};
		\node [style=none] (13) at (-9.25, 0.25) {};
		\node [style=box, minimum height=1cm] (14) at (-11, 0.25) {$L_f$};
		\node [style=none] (23) at (-11.25, 1) {};
		\node [style=Z dot] (24) at (-12.25, -0.5) {};
		\node [style=none] (25) at (-11.25, -0.5) {};
		\node [style=none] (26) at (-12.25, 0.5) {$\vdots$};
		\node [style=none] (27) at (-9.25, -1.5) {};
		\node [style=none] (28) at (-7.75, -1.5) {};
		\node [style=X dot] (29) at (-8.5, -1.5) {};
		\node [style=Z phase dot] (30) at (-9.5, 0.25) {$\ -\frac\pi2~$};
		\node [style=none] (32) at (-4, 0.25) {};
		\node [style=none] (38) at (-3.5, -1.5) {};
		\node [style=none] (39) at (-2, -1.5) {};
		\node [style=X dot] (40) at (-2.75, -1.5) {};
		\node [style=Z phase dot] (41) at (-2.75, -0.75) {$\ -\frac\pi2~$};
		\node [style=none] (42) at (-7, 0) {$\propto$};
		\node [style=X dot] (43) at (-5.75, 0.25) {};
		\node [style=box] (44) at (-4.25, 0.25) {$Y_\alpha$};
		\node [style=none] (45) at (-1.25, 0) {=};
		\node [style=none] (46) at (3.5, 0.25) {};
		\node [style=none] (47) at (4, -1.5) {};
		\node [style=none] (48) at (5.5, -1.5) {};
		\node [style=X dot] (49) at (4.75, -1.5) {};
		\node [style=Z phase dot] (50) at (4.75, -0.75) {$\ -\frac\pi2~$};
		\node [style=X dot] (51) at (0, 0.25) {};
		\node [style=Z phase dot] (52) at (1, 0.25) {$\ -\frac\pi2~$};
		\node [style=X phase dot] (53) at (2.25, 0.25) {$\alpha$};
		\node [style=Z phase dot] (54) at (3.25, 0.25) {$\frac\pi2$};
		\node [style=none] (55) at (6.75, 0) {=};
		\node [style=none] (56) at (8.75, 0.25) {};
		\node [style=none] (57) at (9, -1.5) {};
		\node [style=none] (58) at (10.5, -1.5) {};
		\node [style=X dot] (59) at (9.75, -1.5) {};
		\node [style=X phase dot] (63) at (8, 0.25) {$\alpha$};
		\node [style=none] (64) at (11.5, 0) {=};
		\node [style=none] (66) at (12.5, 0) {};
		\node [style=none] (67) at (14.5, 0) {};
		\node [style=X phase dot] (69) at (13.5, 0) {$\alpha$};
	\end{pgfonlayer}
	\begin{pgfonlayer}{edgelayer}
		\draw (12) to (23.center);
		\draw (24) to (25.center);
		\draw (14) to (13.center);
		\draw (27.center) to (28.center);
		\draw (38.center) to (39.center);
		\draw [in=90, out=0, looseness=1.25] (32.center) to (41);
		\draw (41) to (40);
		\draw (43) to (32.center);
		\draw (47.center) to (48.center);
		\draw [in=90, out=0, looseness=1.25] (46.center) to (50);
		\draw (50) to (49);
		\draw (51) to (46.center);
		\draw (57.center) to (58.center);
		\draw [in=90, out=0] (56.center) to (59);
		\draw (63) to (56.center);
		\draw (66.center) to (67.center);
		\draw [in=90, out=0] (13.center) to (29);
	\end{pgfonlayer}
\end{tikzpicture}
\end{equation}
In the above, we use the relation $Y_\alpha = Z_{-\frac\pi2} X_\alpha Z_{\frac\pi2}$, and some simple ZX-calculus rewrites.
Hence, the diagram of Eq.~\eqref{eq:unitary-sat-zx} is proportional to an $X_\alpha$ rotation where $\alpha$ is uniquely determined by the number of solutions to $f$.
Note that this operation can be easily represented (with at most three gates) using a gate-set such as $\{H, Z_\alpha, \CNOT\}$, in which the set of values of allowed angles $\alpha$ include those that may arise in the diagram of Eq.~\eqref{eq:unitary-sat-zx2} for some number of solutions $N_1$ to the formula $f$; such an operation will be representable using other gate-sets as well.%
	\footnote{%
		Note that the gate-set described here cannot be a single, finite gate set for all values of $n$.
		However, the angles $\alpha$ arising out of instances of satisfiability in this way can be specified in $O(n)$ bits, precisely by characterising them in the way that we have as being related to some integer ranging in $\{0,1,\ldots,2^n\}$ via inverse trigonometric functions.
		For remarks on what can be achieved with finite gate-sets, the reader may be interested in our remarks on the related problem \ACE, in Section~\ref{sec:variations}.
	}

\begin{theorem}\label{thm:post-selection-hardness}
	 $\CE$ is \sharpP-hard.
\end{theorem}
\begin{proof}
	\numberSAT is a \sharpP-complete problem, so it suffices to show that we can count the number of solutions to a Boolean formula using a call to an oracle which solves \CE. Given a Boolean formula $f:\{0,1\}^n\to \{0,1\}$ with $\mathrm{poly}(n)$ terms, construct the diagram of Eq.~\eqref{eq:unitary-sat-zx}. The diagram here for $L_f$ uses $\mathrm{poly}(n)$ of the diagrammatic gadgets for NOT and AND, and hence the complete diagram consists of $\mathrm{poly}(n)$ spiders, each of which may be restricted to having at most $3$ wires.
	We may apply the \CE oracle on this diagram subject to a suitable gate set that can exactly generate the possible X-rotations $X_\alpha$ which may arise. 
  As $C$ is a single-qubit circuit with at most $\mathrm{poly}(n)$  gates, we can calculate the unitary it implements, up to any required precision $2^{-O(\mathrm{poly}(n))}$, in polynomial time.
	We know that the operation realised is of the form $X_\alpha$, so to determine the value of $N_1$, it suffices to estimate the entries of the resulting $X_\alpha$ to within an error of $\tfrac{1}{2\sqrt{2}}$.
	Determining the value of $a_1 / \sqrt{a_0^2 +a_1^2}$ to $2n$ bits of precision is more than sufficient to do this.
\end{proof}

\begin{corollary}
  If there is a polynomial time algorithm for $\CE$, then $\P = \P^{\sharpP}$. In particular, the polynomial hierarchy collapses to the first level: $\P = \NP = \PH$.
\end{corollary}
\begin{proof}
  If $\CE$ can be done in polynomial time, then the above shows that we can solve $\numberSAT$ in polynomial time, and hence $\NP \subseteq \P^\sharpP = \P$.
\end{proof}

\begin{remark}
  Our construction of the diagram we use to prove our result might seem somewhat arbitrary. To motivate it some more, first realise that instead of the function $L_f$, we could have used the standard unitary quantum oracle for a Boolean function $U_f$ which acts on $n+1$ qubits via $U_f\ket{x, b} = \ket{x, b\oplus f(x)}$. We can get $L_f$ out of $U_f$ by post-selecting the top $n$ qubits to $\bra{+}$. Using the language of post-selection, we may then present a circuit version of Eq.~\eqref{eq:unitary-sat-zx}:
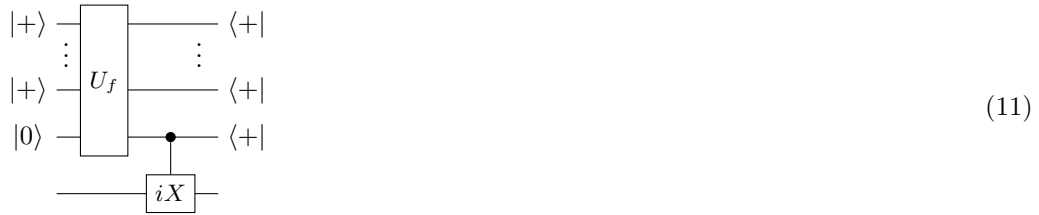
  \begin{equation}\label{eq:post-selection-SAT-circuit}
  \begin{tikzpicture}
	\begin{pgfonlayer}{nodelayer}
		\node [style=box, minimum height=2.0cm] (0) at (-1.75, 0.75) {$U_f$};
		\node [style=none] (1) at (-3, 2.25) {};
		\node [style=none] (2) at (1.25, 2.25) {};
		\node [style=none] (3) at (-3, 0.5) {};
		\node [style=none] (4) at (1.25, 0.5) {};
		\node [style=none] (5) at (-3, -0.75) {};
		\node [style=none] (6) at (1.25, -0.75) {};
		\node [style=none] (7) at (-3.75, 2.25) {$\ket{+}$};
		\node [style=none] (8) at (-3.75, 0.5) {$\ket{+}$};
		\node [style=none] (9) at (-2.75, 1.625) {$\vdots$};
		\node [style=none] (10) at (-3.75, -0.75) {$\ket{0}$};
		\node [style=cnot ctrl] (11) at (0, -0.75) {};
		\node [style=none] (12) at (-3, -2.25) {};
		\node [style=none] (13) at (1.25, -2.25) {};
		\node [style=box] (14) at (0, -2.25) {$iX$};
		\node [style=none] (15) at (2, 2.25) {$\bra{+}$};
		\node [style=none] (16) at (2, 0.5) {$\bra{+}$};
		\node [style=none] (17) at (2, -0.75) {$\bra{+}$};
		\node [style=none] (18) at (0.75, 1.625) {$\vdots$};
	\end{pgfonlayer}
	\begin{pgfonlayer}{edgelayer}
		\draw (2.center) to (1.center);
		\draw (4.center) to (3.center);
		\draw (6.center) to (5.center);
		\draw (11) to (14);
		\draw (12.center) to (13.center);
	\end{pgfonlayer}
\end{tikzpicture}
  \end{equation}
  The top part is calculating the number of solutions, while the bottom part ensures that this information is fed into a qubit in such a way that the overall operation is proportional to a unitary.
  The choice of $iX$ is for the sake of simplicity: any unitary $U$ that satisfies $U = -U^\dagger	$ would also suffice, such as $iY$ or $iZ$.
\end{remark}

\begin{remark}
  Even though we can view the diagram as a post-selected circuit, this does not seem to be where the power of the procedure comes from, as it is for instance in Aaronson's~\cite{aaronson2005quantum} characterisation \textbf{PostBQP = PP}.
  The probability of observing the `correct' outcome is bounded from below by a constant, and does not depend on $n$. This means in particular that by doing repeat-until-success we could with high probability implement the circuit Eq.~\eqref{eq:post-selection-SAT-circuit} on a quantum computer. However, this does not allow you to solve \numberSAT, as adjacent possibilities of the rotation angle $\alpha$ are exponentially close. So rather, the power of the procedure comes from getting an explicit description of the circuit which allows us to exactly calculate the rotation angle.
\end{remark}

\section{Variations on extraction}
\label{sec:variations}

There are several variations on circuit extraction which we can consider, all of which also turn out to be hard.

The essential trick we used in our proof is that our resulting circuit has just \emph{one} qubit, and hence a description of a unitary on it can easily be transformed into the actual unitary it implements by just multiplying all the resulting matrices. But of course the same statement remains true if we have slightly more than one qubit, say a logarithmic amount in the size of the \textbf{SAT} instance. We also see that it then doesn't matter if our circuit contains auxiliary qubits, measurements, or classically-controlled corrections. All of these can be efficiently calculated as long as the number of qubits is small enough. Therefore, let's define the following variant of circuit extraction.

\begin{quote}
\noindent \textbf{AuxCircuitExtraction} \\
\noindent \textbf{Input}: A ZX-diagram $D$ with $n$ inputs and outputs and at most $k$ wires and/or spiders, and a set $\mathcal G$ of unitary gates (each acting on at most $O(1)$ qubits). \\
\noindent \textbf{Promise}: The operator denoted by $D$ is proportional to a unitary. \\
\noindent \textbf{Output}: Either \textbf{(a)}~a~deterministic $n$-qubit circuit implementing the unitary of the input ZX-diagram, described as a $\text{poly}(n,k)$ length list of gates, auxiliary qubit preparations, measurements, and classical corrections, with at most $O(\log k)$ auxiliary qubits; or \textbf{(b)}~a~message that no such circuit exists, if that is the case.
\end{quote}

\begin{theorem}
	\textbf{\upshape AuxCircuitExtraction} is \sharpP-hard.
\end{theorem}
\begin{proof}
	We construct the same diagram as in the proof of Theorem~\ref{thm:post-selection-hardness} to solve a \numberSAT instance, except that we can no longer assume that the final circuit will act only on a single qubit: instead it may act on up to $O(\log k)$ qubits, including the operations on the auxiliary qubits.
	The size of the matrices involved when trying to calculate the resulting unitary is $O(2^{\log \text{poly}(k)}) = O(\text{poly}(k))$, where here $k$ is the size of the input diagram.
	We may then still multiply the matrices together in polynomial time to obtain sufficiently precise estimates of the coefficients.
\end{proof}

One might also object that requiring the output unitary to \emph{exactly} represent the ZX-diagram is too strong --- in particular, impossible in general even with an approximately universal, finite gate set --- and wish for an approximate output instead.
We say that a unitary operator $\tilde U$ is an \emph{$\varepsilon$-approximation} of another unitary $U$ for some $\varepsilon > 0$, if $\lVert \tilde U - e^{i\alpha} U \rVert < \varepsilon$ for some global phase $\alpha$.
Here, $\lVert M \rVert$ denotes the operator norm of $M$: the largest singular value of $M$.
\begin{quote}
\noindent $\ACE$ \\
\noindent \textbf{Input}: A ZX-diagram $D$ with $n$ inputs and outputs and at most $k$ wires and/or spiders, a set $\mathcal G$ of unitary gates (each acting on at most $O(1)$ qubits), and a precision parameter $\varepsilon > 0$. \\
\noindent \textbf{Promise}: The operator denoted by $D$ is proportional to a unitary. \\
\noindent \textbf{Output}: Either \textbf{(a)}~a~$\text{poly}(n,k,\log(1/\varepsilon))$-size circuit $C$, expressed as a sequence of gates from $\mathcal G$  and expressing an $n$-qubit unitary $\tilde U$ which is an $\varepsilon$-approximation to either the operator denoted by $D$, or some operator proportional to it; or \textbf{(b)}~a~message that no such circuit exists, if that is the case.
\end{quote}

\begin{theorem}
  $\ACE$ is \sharpP-hard.
\end{theorem}
\begin{proof}
  For a given \textbf{SAT} instance $f:\{0,1\}^n\to\{0,1\}$ we again construct the same diagram as in the proof of Theorem~\ref{thm:post-selection-hardness} which denotes a unitary $X_\alpha$\,, where $\alpha$ allows us to determine the number of solutions to $f$. This diagram has $\text{poly}(n)$ spiders. Set $\varepsilon = 2^{-cn}$ for some large enough constant $c$. Then applying \ACE gives rise to a circuit, which has $\text{poly}(\text{poly}(n),\log(1/2^{-cn})) = \text{poly}(n)$ gates. We can hence just multiply out the matrices in order to determine the unitary $U$ it implements. This unitary $U$ approximates $X_\alpha$ to degree $2^{-cn}$. Since the top left entry of $X_\alpha$ is real, we can first multiply $U$ by the appropriate global phase to ensure it is also real. If we have picked $c$ large enough then the entries of $U$ are then within $\frac12 2^{-n}$ of that of $X_\alpha$ so that we can determine $\alpha$ by rounding to the nearest allowed value.
\end{proof}
Note that, even for exponentially small angles $\alpha$ as might arise when $f$ has few solutions, circuits of polynomial size do \emph{exist} for $X_\alpha$ when $\mathcal G$ is an approximately universal gate-set: using the Solovay--Kitaev algorithm~\cite{Kitaev-1997,DawsonN-2005} or any of its many refinements (see \emph{e.g.}~Ref.~\cite{Bouland-SK-2021} and references therein), we may synthesise circuits approximating $X_\alpha$ to any precision $\varepsilon$ in time scaling polynomially in $\log(1/\varepsilon)$.
The difficulty of \ACE\ stems from determining \emph{which} angle $\alpha$ to approximate.

Let us consider one final variation on extraction.
One could argue that the reason that we end up with a hard problem in these instances, is because requiring the output to be some kind of circuit is too restrictive. The ultimate goal of circuit extraction is that we wish for the ZX-diagram to be run on a quantum computer, to obtain some probability distribution over outcomes; but the complexity of \CE\ and its variations seems to arise from the complexity of finding a precise description of the procedure to do so.
Cutting out the middle-man, we may consider \emph{any} process which takes as input a unitary ZX-diagram, and produces bit strings as output whose distribution conforms with the one we expect from the unitary.

\begin{quote}
\noindent \textbf{UnitaryZXSampling} \\
\noindent \textbf{Input}: A ZX-diagram $D$ with $n$ inputs and outputs and at most $k$ wires and/or spiders. \\
\noindent \textbf{Promise}: The operator denoted by $D$ is proportional to some unitary $U$. \\
\noindent \textbf{Output}: A sample $x \in \{\texttt0,\texttt1\}^n$ from a probability distribution, given by (or sufficiently close to) $\lvert \bra{x} U \ket{\texttt{0$\cdots$0}}\rvert^2$.
\end{quote}

It is clear that \textbf{UnitaryZXSampling} is at least as hard as $\BQP$: we could just input a ZX-diagram that directly represents a quantum circuit, in which case this problem is equivalent to simulating that circuit.
The reason we write here that the probabilities just have to be `sufficiently close' is because the exact number doesn't matter for the theorem below. To be concrete we could for instance allow the probability to additively deviate by~$1/3$ from the true value.

\begin{theorem}\label{thm:unitaryZXsampling}
  There is a randomised polynomial reduction from \NP to \textbf{\upshape UnitaryZXSampling}. In other words: with access to a \textbf{\upshape PromiseUnitaryZXSampling} oracle---which produces the expected output if the input diagram is unitary and arbitrary output otherwise---we can with high probability solve \NP-complete problems.
\end{theorem}
\begin{proof}
  \textbf{SAT} is an \NP-complete problem. To randomly reduce \NP it however suffices to consider the problem \textbf{USAT} by the Valiant–Vazirani theorem~\cite{ValiantVazirani1985}. \textbf{USAT} asks us to determine whether a Boolean formula is satisfiable, given the promise that it has at most one solution.
  Using the randomised reduction from \textbf{SAT} to \textbf{USAT}, we consider how to solve \textbf{USAT} using a \textbf{\upshape PromiseUnitaryZXSampling} oracle.
  
  Let $f:\{0,1\}^n\to \{0,1\}$ be a Boolean formula that has at most one solution. Construct the diagram Eq.~\eqref{eq:unitary-sat-zx} as in the previous proofs: as a unitary this implements the identity iff $f$ is not satisfiable, and $X_\alpha$ for some fixed angle $\alpha > 0$ when $f$ is satisfiable.
  In the latter case, the value of $\alpha$ is exponentially small, but known precisely, as $f$ has exactly one solution in this case. So we can say the circuit implements $X_{s\cdot \alpha}$ where $s\in\{0,1\}$ encodes whether $f$ is satisfiable or not.

  Let $M$ be the one-qubit (non-unitary) matrix that maps $\ket{\texttt0} \mapsto \ket{\texttt0}$ and $X_\alpha\ket{\texttt0} \mapsto \ket{\texttt1}$, so that in particular $M X_{s\cdot \alpha} \ket{\texttt0} = \ket{s}$.  
  By universality of ZX-diagrams we can find some (constant sized) diagram to represent $M$.
  We can then calculate:
  \begin{equation}\label{eq:ZX-sampling-circuit}
    \begin{tikzpicture}
	\begin{pgfonlayer}{nodelayer}
		\node [style=Z dot] (77) at (-9, 3.5) {};
		\node [style=none] (78) at (-6.5, 2.75) {};
		\node [style=box, minimum height=1cm] (79) at (-7.75, 2.75) {$L_f$};
		\node [style=none] (80) at (-8, 3.5) {};
		\node [style=Z dot] (81) at (-9, 2) {};
		\node [style=none] (82) at (-8, 2) {};
		\node [style=none] (83) at (-9, 3) {$\vdots$};
		\node [style=X dot] (84) at (-6.5, 0.5) {};
		\node [style=Z dot] (85) at (-1.75, 0.5) {};
		\node [style=X dot] (86) at (-5.25, 0.5) {};
		\node [style=Z phase dot] (87) at (-5.25, 1.75) {$\ -\frac\pi2~$};
		\node [style=box] (88) at (-4.25, 0.5) {$M$};
		\node [style=none] (89) at (-4.5, -0.5) {};
		\node [style=none] (90) at (-1.75, -0.5) {};
		\node [style=Z dot] (91) at (-2.75, 0.5) {};
		\node [style=X dot] (92) at (-2.75, -0.5) {};
		\node [style=X dot] (100) at (1.5, 0.75) {};
		\node [style=Z dot] (101) at (6.5, 0.75) {};
		\node [style=X phase dot] (102) at (3, 0.75) {$\ s\cdot\alpha~$};
		\node [style=box] (104) at (4.5, 0.75) {$M$};
		\node [style=none] (105) at (3.75, -0.5) {};
		\node [style=none] (106) at (6.5, -0.5) {};
		\node [style=Z dot] (107) at (5.5, 0.75) {};
		\node [style=X dot] (108) at (5.5, -0.5) {};
		\node [style=none] (109) at (0, 0) {=};
		\node [style=none] (110) at (0, 1) {\eqref{eq:unitary-sat-zx2}};
		\node [style=Z dot] (112) at (12, 0.75) {};
		\node [style=X phase dot] (113) at (9.5, 0.75) {$\ s\pi~$};
		\node [style=none] (115) at (9.25, -0.5) {};
		\node [style=none] (116) at (12, -0.5) {};
		\node [style=Z dot] (117) at (11, 0.75) {};
		\node [style=X dot] (118) at (11, -0.5) {};
		\node [style=none] (119) at (7.75, 0) {=};
		\node [style=X phase dot] (122) at (15.25, 0) {$\ s\pi~$};
		\node [style=none] (123) at (14.25, 0) {};
		\node [style=none] (124) at (16.25, 0) {};
		\node [style=none] (127) at (13, 0) {=};
	\end{pgfonlayer}
	\begin{pgfonlayer}{edgelayer}
		\draw (77) to (80.center);
		\draw (81) to (82.center);
		\draw (79) to (78.center);
		\draw (84) to (85);
		\draw [in=90, out=0, looseness=1.25] (78.center) to (87);
		\draw (87) to (86);
		\draw (91) to (92);
		\draw (92) to (90.center);
		\draw (92) to (89.center);
		\draw (100) to (101);
		\draw (107) to (108);
		\draw (108) to (106.center);
		\draw (108) to (105.center);
		\draw (117) to (118);
		\draw (118) to (116.center);
		\draw (118) to (115.center);
		\draw (113) to (117);
		\draw (112) to (117);
		\draw (122) to (124.center);
		\draw (122) to (123.center);
	\end{pgfonlayer}
\end{tikzpicture}
  \end{equation}
  Hence, the ZX-diagram on the left in Eq.~\eqref{eq:ZX-sampling-circuit}  implements either the identity, or an $X_\pi$ operation (that is to say, a NOT operation), depending on whether $f$ is satisfiable. When we feed this ZX-diagram to an oracle for \textbf{PromiseUnitaryZXSampling}, we get either the output $\texttt0$ or $\texttt1$, where a $\texttt0$ indicates with high probability that the circuit is the identity, and a $\texttt1$ indicates that the circuit is a NOT operation. We can repeatedly call the oracle to get additional samples to increase our confidence in the result.
  
  Now suppose $f$ is a general instance of $\mathbf{SAT}$, which may have more than one solution. 
  Using the Valiant--Vazirani reduction multiple times we probabilistically produce different Boolean formulae $f_1,\ldots, f_m$. If $f$ is not satisfiable, then none of the $f_j$ will be satisfiable either and this is what the \textbf{PromiseUnitaryZXSampling} will tell us as well. 
  If $f$ \emph{is} satisfiable, then a significant fraction of the $f_j$ will have a unique solution, so that our oracle tells us they are satisfiable. For the other $f_j$ the oracle will return some arbitrary output. So by picking $m$ large enough there will with high probability be some $f_j$ that will be uniquely satisfiable, and so we can conclude that $f$ is satisfiable as well.
  
  Hence, we can determine with arbitrary high probability whether a \textbf{SAT} instance is satisfiable using enough calls to \textbf{PromiseUnitaryZXSampling}.
\end{proof}

\begin{remark}
  If we knew that the number of solutions to the \textbf{SAT} instance was some other fixed number, then we could pick a different matrix $M'$ to boost the state up to $X_\pi$ gate as above. If we pick $M'$ `slightly wrong', then the resulting diagram will just be close to $X_\pi$. One might think that we could use such a procedure to try and determine the number of solutions to $f$ by doing binary search on the number of solutions, and so boost the power of \textbf{UnitaryZXSampling} to \sharpP. However, the problem with this is that the resulting diagrams are not proportional to a unitary most of the time. There might be some way around this issue, so that \textbf{UnitaryZXSampling} is still \sharpP-hard: we leave this as an open problem.
\end{remark}

\begin{remark}
Note that if we were to consider a version of \textbf{UnitaryZXSampling}, without the promise of unitarity, such an oracle would be as powerful as \PostBQP, since we can represent any ZX-diagram as a post-selected quantum circuit.
In our case, the power again comes not so much from postselection, as being able to take advantage of the versatility of ZX-diagrams to gain access, in some way, to extract very precise information regarding a \sharpP\ problem. 
\end{remark}

\section{Upper bounding the complexity of \CE}\label{sec:upperbound}

Given that \CE is \sharpP-hard, one might ask whether or not the problem is \sharpP-complete (or more precisely: $\mathbf{FP}^\sharpP$-complete), in the sense that a Turing machine with access to a \sharpP oracle would be able to solve it, for some given polynomial upper bound on circuit size and some given gate-set (perhaps with suitable restrictions), in polynomial time.
We have not managed to prove such a completeness result. We will however present the following upper bounds on decision problem versions of circuit extraction, relying on techniques from~\cite{Adleman-DeMarrais-Huang1997} (which we describe in Appendix~\ref{apx:countingComplexityQm}).

First, consider the following decision problem: given a ZX-diagram, and a circuit, determine whether the circuit implements a unitary which is proportional to that represented by the ZX-diagram (whether by a factor of $\mathrm{e}^{i\theta}$ for some angle $\theta$, or a more general complex number).
The complement of this problem is in $\textbf{NP}^\sharpP$.
To sketch why this is, consider a circuit $C$ representing a unitary $U$, and a ZX diagram $D$ representing an operator $V$.
If $a_0$ and $a_1$ are two non-zero coefficients from $U$, and $b_0$ and $b_1$ are the corresponding (non-zero) pair of coefficients from the matrix $U$ represented by $C$, then $U \propto V$ only if either $a_0/b_0 = a_1/b_1$ for all possible such pairs.
(We compare pairs instead of single coefficients, to deal with the fact that they might only represent the same matrix up to some non-zero scalar.)
We also require that for any coefficient $a$ in $U$ which is zero, the corresponding coefficient $b$ of $V$ is also zero.
Taken together, this implies that for all corresponding pairs of coefficients of $U$ and $V$, $a_0 b_1 = a_1 b_0$.
A \sharpP oracle allows one to calculate coefficients%
	\footnote{%
		In this case, it is not necessary to compute complete information about $a_0$, $a_1$, $b_0$, and $b_1$: it suffices to compute information about individual components of the products $a_0 b_1$ and $a_1 b_0$, where these are regarded as vectors over $\mathbb Q$.
		(See the closing remarks of Appendix~\ref{apx:countingComplexityQm}.)
	} 
of ZX-diagrams and circuits: if $U \not\propto V$, an NTM with access to a \sharpP oracle can non-deterministically find a witness that these two operators are not in fact proportional to one another.
Thus, determining whether a circuit does \emph{not} represent a unitary which is denoted (up to scalar factors) by a ZX-diagram, is in $\NP^\sharpP$.

The above result has a simple corollary: the property of a ZX-diagram being proportional to a unitary, itself belongs to the complement of $\NP^\sharpP$.
We may see this by the fact that a ZX-diagram denoting an operator $V$, which is proportional to a unitary, satisfies $VV^\dagger \propto I$.
We may represent $VV^\dagger$ by composing the diagram $D$ with its adjoint (which is the left-to-right mirror image of $D$, with all phase angles negated).
This composite diagram may easily be computed, at which point we may ask whether the operator it represents is distinct from the identity (even modulo non-zero scalar factors).
As we note above, this problem is in $\NP^\sharpP$.

Finally, using these ideas, we may consider the decision problem of determining, for a ZX-diagram $D$ which denotes an operator $V$ proportional to a unitary,  and given some (approximately universal) gate set $\mathcal G$ and polynomial length bound $N$, whether \emph{there exists} a circuit of at most $N$ gates over $\mathcal G$ which implements $V$.
This problem is in $\NP^{\NP^\sharpP}$: for an NTM with access to an $\NP^\sharpP$ oracle, it suffices to make a nondeterministic guess at a circuit of length $N$ (where each gate may be the identity operator, or some gate $G \in \mathcal G$ acting on a non-deterministically chosen set of qubits) and then query the oracle to determine whether the circuit approximately realises $V$.
A deterministic Turing machine, with access to an oracle for this problem, could then solve \ACE\ in polynomial time using standard techniques, using the oracle to facilitate a search for a circuit to realise $D$.

These observations represent the most straightforward approach to determining an upper bound for the circuit extraction problem, and seem to place it at a level of complexity significantly higher than $\P^\sharpP$.
If we conceive of $\sharpP$ as broadly representing the complexity of evaluating a tensor network, a superficial analogy between \CE\ and boolean circuit minimisation~\cite{garey-johnson-1979,Buchfuhrer-2011} would seem to suggest that \CE\ is likely to be hard for some complexity class higher than $\P^\sharpP$ (barring some collapse of complexity classes).

\section{Conclusion}\label{sec:conclusion}

In this paper we studied the problem of extracting a quantum circuit description from a unitary ZX-diagram. We've shown that this problem is \sharpP-hard by reducing \numberSAT to an application of circuit extraction. We've also studied some variations where we allow auxiliary qubits, classical control, and/or approximate synthesis of the desired unitary, and have shown that these problems are also \sharpP-hard. In addition, we studied the hardness of a machine that takes in a unitary ZX-diagram and outputs measurement samples from that ZX-diagram, and have shown that such a machine allows one to probabilistically solve \NP-hard problems.

A conclusion to be drawn from our results is that if we want some efficient procedure to transform a unitary ZX-diagram into a quantum circuit, then we will have to have some additional information about the structure of the ZX-diagram.
In the known procedures for efficient circuit extraction~\cite{Backens2020extraction,Simmons2021Measurement,deBeaudrap2020Paulifusion}, this additional information takes the form of a kind of `flow' on the diagram that prevents parts of the diagram from becoming too unwieldy.
An immediate question then is if there are other types of, more general, promises on the structure of the diagram which then allow you to extract a circuit from it.

Aaronson showed that sampling from a post-selected quantum circuit is hard~\cite{aaronson2005quantum}. Our results imply that some other tasks surrounding \emph{unitary} post-selected circuits (that is, circuits which perform a unitary transformation conditioned on some post-selection) are hard. 
However, this hardness seems to stem not from the post-selection itself, as the post-selections can be simulated with high probability in our case.
Rather, the hardness seems to stem from a hypothetical ability to find an equivalent, deterministic way to realise the same operation --- which implies an ability to extract difficult-to-access information about the input diagram.

A question related to circuit extraction from ZX-diagrams is circuit extraction from deterministic measurement patterns (in for instance the one-way model or lattice surgery). When we have a deterministic measurement pattern, we can represent each branch of the computation by a ZX-diagram denoting a unitary.
Our hardness proof does however not immediately translate to this setting, as it might be that the fact that all of these ZX-diagrams are branches of the same measurement pattern forces some kind of structure on the diagrams that might make it easier to rewrite them into circuits. The diagrams we used to show hardness of circuit extraction are as far as we are aware not representable as branches of some deterministic measurement pattern, so that we can't use the same proof.
We leave it for future work to determine the hardness of extracting unitary circuits from deterministic measurement patterns.

\acknowledgements{JvdW is supported by an NWO Rubicon personal fellowship.}

\bibliography{main}

\appendix
\section{Quantum circuits as a computational model}
\label{apx:circuits}

In order to describe a model of computation with bounded computational power, one usually imposes further constraints on `quantum circuits', as follows.

As with the study of boolean circuits as a model of computation, one often considers a quantum circuit to be described by some polynomial-time computable procedure (a sort of `effective blueprint'), which for a given $n \ge 0$ requires time $\mathrm{poly}(\log n)$ to produce a circuit taking inputs of size $n$.
While this constraint is not essential when considering a single circuit on its own (the description of the circuit itself is a finite specification), this constraint prevents us from considering what might otherwise seem like `quantum algorithms' for uncomputable problems (in the same way that one must for boolean circuits).

Furthermore, to prevent unbounded computational power from being hidden elsewhere in the description of a quantum circuit, one often imposes constraints on the gates and measurements allowed in a circuit.
One common convention is to require that the unitary gates be drawn from a finite set of unitary operators, and to require that all measurements are of single qubits, with projectors $\{ \ketbra{\texttt{0}}{\texttt{0}}, \ketbra{\texttt{1}}{\texttt{1}} \}$.
This is frequently relaxed, to permit arbitrary single-qubit operations, or Z-rotations and X-rotations by arbitrary angles, and to allow one or two particular two-qubit gates, such as CNOT or CZ, and single-qubit measurements in an arbitrary basis; one then requires that the parameters of any such gates or measurements be efficiently computable.

\section{Counting complexity upper bounds on quantum simulation}
\label{apx:countingComplexityQm}

In order to consider how we may obtain containments for quantum computing using counting complexity, we give an outline of the known techniques for relating quantum computation to counting complexity, and the (mild) technical constraints on coefficients of gates and of ZX-diagrams which this involves.

We are frequently interested in unitary circuits whose gates have coefficients which are irrational.
However, in practise, the (finite) gate-sets of most interest involve only algebraic coefficients: that is to say, complex numbers $z$ (\emph{e.g.},~$\smash{\tfrac{1}{\sqrt 2}}$ and $\mathrm e^{2\pi i/k}$ for integers $k>0$) which satisfy $q(z) = 0$, for some polynomial $q(x) = a_d x^d + \cdots + a_1 x + a_0$ with integer coefficients.
As such coefficients are dense in the complex numbers, and in particular are also dense in the subset of the complex numbers which have unit modulus, coefficients of this sort are appropriate to (approximately) represent gates which we may normally think of as being drawn from continuous, single-parameter families of gates such as $Z_\theta$ and $X_\theta$.
This motivates a restriction to gates, whose coefficients are represented by algebraic numbers.

We refine this constraint as follows.
By virtue of a circuit having only finitely many gates $N \ge 0$, the coefficients of the gates can in principle be drawn from a common algebraic number field $K = \mathbb Q[\omega_1, \omega_2, \cdots, \omega_m]$, where each $\omega_k$ is a root of some (irreducible) polynomial $q_j(x)$ as above and where $m \le N$.
(For a finite gate-set, $m$ will in fact be a constant independent of $N$.)
We adopt the requirement on the representation of a circuit, that the coefficients for its gates be explicitly represented in a form
\begin{equation*}
	 \frac{1}{M}
	 \;\;\sum_{\mathclap{p_1, \ldots\!\!\;,\!\: p_m}}\;\;
	 	 a_{p_1, \ldots\!\!\;,\, p_m} \omega_1^{\,p_1} \cdots \omega_m^{\,p_m}
\end{equation*}
for some constant $M$ and for some choice of integer coefficients $a_{p_1,\ldots\!\!\;,\,p_m}$, where $0 \le p_j < \deg(q_j)$ for each $j$.
(Following Ref.~\cite{Adleman-DeMarrais-Huang1997}, for a finite gate-set, this does not limit the computational power of unitary circuits either for exact or for bounded-error quantum computation.)
A product of $N$ such coefficients, may be represented in a similar form, but with a leading factor of $1/M^N$ rather than $1/M$; sums of such products may be represented in the same way by collecting terms together.
For two numbers represented in this form, we may put them into normal form by comparing each of the terms $\omega_1^{\,p_1} \cdots \omega_m^{\,p_m}$ (for some fixed sequence of powers $p_1,\ldots,p_m$), and testing equality of each integer coefficient.
The significance of this, is that a circuit with $N$ gates, where the individual gates satisfy such conditions, represent unitary transformations whose coefficients take this form.
This provides us with a way in which the coefficients of operators given as unitary circuits, may be represented in such a way that we can efficiently perform equality comparisons \emph{in principle};%
	\footnote{%
		It is in principle possible to extend this representation in some ways, to allow equality comparisons between coefficients which are not algebraic; but this is not practically important to us.
	}
the question is then how we bound the complexity of doing so.

Using the techniques of Ref.~\cite{Adleman-DeMarrais-Huang1997}, we may represent the coefficients of a unitary operator indirectly, using the number of accepting paths of nondeterministic Turing machines with various strings written on the tape to represent the integers $a_{p_1,\cdots\!\!\;,\,p_m}$ of an algebraic number as above.%
	\footnote{%
		To be more precise: we may represent such an integer $a_{p_1,\cdots\!\!\;,\,p_m}$ as a difference $a^{(+)}_{p_1,\cdots\!\!\;,\,p_m} - a^{(-)}_{p_1,\cdots\!\!\;,\,p_m}$, where each $a^{(\pm)}_{p_1,\cdots\!\!\;,\,p_m}$ is itself the number of accepting branches on the nondeterministic Turing machine, with slightly different contents of its tape.$\big.$}
Then, with access to a \sharpP oracle, a deterministic Turing machine would be able to evaluate any one coefficient $a_{p_1,\cdots\!\!\;,\,p_m}$.
Furthermore: the techniques of Ref.~\cite{Adleman-DeMarrais-Huang1997} are not specialised in any particular way to unitary matrices, or indeed to square matrices, and can be applied to a matrix of any shape (or indeed, any tensor network) with algebraic coefficients.
In particular, we may use these techniques to similarly represent the coefficients of linear operators which are represented by ZX-diagrams --- assuming again that the coefficients arising from each generator is algebraic.
The coefficients of the operators described in Eq.~\eqref{eqn:spiders} are either $0$, $1$, or of the form $\mathrm e^{i\theta}$ or $2^{-k/2} (1 \pm \mathrm{e}^{i\theta})$ for various phase angles $\theta$ and integers $k$.
For these to all be algebraic, it suffices that the phase angles $\theta$ of the ZX-diagrams are rational multiples of $\pi$ (though other values of $\theta$ would also satisfy this constraint).%
	\footnote{%
		Technically, for any rational multiple $a\pi/d$ described in the diagram, we also require that a prime factorisation of $d$ is provided.
		(This is an implicit requirement for the representation described above of elements of a number field $K$.)
		For the purposes of practical numerical computation, this can be fulfilled by taking $d = 2^N$ for a sufficiently large value of $N$.
		Alternatively, if one were to use such a representation within a model of computation with access to a \sharpP oracle, one might simply factorise $d$ by making use of the \sharpP oracle to facilitate a binary search for factors of $d$.
	}
We may also consider the coefficients which arise from multiplying the matrix represented by a ZX-diagram, by the matrix represented by a unitary circuit; this requires only that we represent the coefficients of each in a consistent way in terms of a common algebraic number field $K$.
From representations of the gates and of the coefficients of the ZX-generators, this would not be difficult to compute.

The above suffices to describe, in outline, how we may describe the coefficients of a linear operator $U$ --- described by a unitary circuit, a ZX-diagram, or compositions of these representations --- by counting the number of accepting branches of non-deterministic Turing machines.
In particular, for any one coefficient $M^{-N} \sum_{p_1,\cdots\!\!\:,\,p_m} a_{p_1,\ldots\!\!\;,\,p_m} \omega_1^{\,p_1} \cdots \omega_m^{\,p_m}$, a deterministic Turing machine with a \sharpP oracle can evaluate any particular integer $a_{p_1,\ldots\!\!\;,\,p_m}$.
If the degree%
	\footnote{%
		A number field $K$ such as we have considered, can be interpreted as a vector space over the rational numbers $\mathbb Q$, with a `basis' described by the monomials $\omega_1^{\,p_1} \cdots \omega_m^{\,p_m}$ for various integers $0 \le p_j < \deg(q_j)$.
		The \emph{degree} $D$ is then the dimension of $K$ as a vector space over $\mathbb Q$, and may be computed as the product $D = \deg(q_1) \cdots \deg(q_m)$ of the degrees of the irreducible polynomials $q_j$.
	}
of the field extension $D = [K:\mathbb Q]$ is `small' (\emph{i.e.}, a constant, or more generally bounded by a polynomial in the size of the description of $U$ as a composition of simpler operators), we may simply query each such integer to obtain complete information of a single coefficient of $U$.
(If the degree is not `small' in this sense, one might consider subtler strategies to obtain information about them; absence any particular promises or structure, complete information about the matrix coefficients would likely be inaccessible in polynomial time, even with recourse to a \sharpP oracle.)

\end{document}